\documentclass[a4paper,11pt]{article}
\usepackage{jheppub} 
\usepackage{lineno}
\usepackage[T1]{fontenc} 
\usepackage{braket}
\usepackage{graphicx}
\usepackage{amsthm}
\usepackage{amsmath, amssymb}
\usepackage{bbold}
\usepackage{comment}
\excludecomment{confidential}
\usepackage{makecell}
\usepackage{bm}
\usepackage{mathtools}
\usepackage{graphicx}

\newtheorem{theorem}{Theorem}

\DeclareFontFamily{U}{matha}{\hyphenchar\font45}
\DeclareFontShape{U}{matha}{m}{n}{
      <5> <6> <7> <8> <9> <10> gen * matha
      <10.95> matha10 <12> <14.4> <17.28> <20.74> <24.88> matha12
      }{}

\title{Towards a dual formulation of quantum gravity via metric-curvature bijections}

\author{Praveen D. Xavier}
\affiliation{Rudolf Peierls Centre for Theoretical Physics, Clarendon Laboratory, Parks Road, University of Oxford, Oxford OX1 3PU, UK}

\emailAdd{praveen.xavier@physics.ox.ac.uk}

\abstract{We prove that Riemannian metrics in General Relativity in the \emph{`normal-coordinates'} gauge are in one-to-one correspondence with curvature 2-forms.
We discuss how this can be used as a change of variables in the operator formalism to construct a dual formulation of quantum gravity pertinent in the context of asymptotic safety-like approaches to quantum gravity.
}


\makeatletter
\gdef\@fpheader{}
\makeatother

\begin{document}
\maketitle
\flushbottom

\section{Introduction}

\subsection{Revisiting a classical problem... \label{cps}}
In Yang-Mills (YM) theory, Wu and Yang have shown that there exists connections, unrelated by gauge transformation, that have the same curvature \cite{wu-yang}. This is known as the \emph{field copy problem} \cite{mostow}. 
Halpern pointed out that the curvature map restricted to connections in the \emph{axial gauge} is, however, injective \cite{halpern}.
Later, Durand and Mendel showed this also in the \emph{Fock-Shwinger gauge} \cite{durand}. 
(The Fock-Schwinger gauge is preferred to the axial gauge because of its relative simplicity.)
To be precise, Durand and Mendel showed that connections in the Fock-Shwinger gauge are mapped, by the curvature map, \emph{bijectively} to curvature 2-forms, $F$, satisfying the `YM Bianchi identity for curvature':
\begin{equation}
\begin{split}
    &dF+ig[A\wedge F]=0, \\
    \text{where}\quad&A_{\mu}(x)=\int_0^1 tx^\nu F_{\nu\mu}(tx) dt.
\end{split}
\end{equation}


Just as in YM, in General Relativity, Riemannian metrics unrelated by coordinate transformation may have the same curvature. 
Müller, Schubert and van de Ven have shown, however, that the curvature map restricted to metrics in the \emph{`normal-coordinates'} gauge is \emph{injective} \cite{muller}.
So that we can set up a bijection, we pose the following question: what is the image of this map?
In this paper we prove that the image (of the curvature map restricted to metrics in the `normal coordinates' gauge) is the restricted set of curvature 2-forms, $R^a_{b}$, satisfying what we call the \emph{`1\textsuperscript{st} and 2\textsuperscript{nd} Bianchi identities for curvature'}:
\begin{equation}
\begin{split}
    &R^a_{b} \wedge e^b=0, \\
    \text{where}\quad&e^a_\mu(x)=\delta^a_\mu+\int_0^1 t(1-t)x^b x^\nu R^a_{b\nu\mu}(tx) dt
\end{split}
\end{equation}
and 
\begin{equation}
\begin{split}
    &dR^a_b+\omega^a_c\wedge R^c_b-R^a_c\wedge \omega^c_b=0, \\
    \text{where}\quad&\omega^a_{b\mu}(x)=\int_0^1 tx^\nu R^a_{b\nu\mu}(tx) dt
\end{split}
\end{equation}
respectively. 
This result proves that there is a bijection between metrics in normal coordinates and curvature 2-forms satisfying the 1\textsuperscript{st} and 2\textsuperscript{nd} Bianchi identities for curvature. We discuss, below, how this result can be used to derive a dual formulation of quantum gravity.

\subsection{...and its quantum application}

A quantum theory is usually formulated in terms of operators $\vec q$ (which denotes the collection $\{q_\alpha^i\}$ -- where $\alpha$ may be a continuous index and $i$ may be a discrete index) and their canonical conjugates $\vec p$, with dynamics governed by a Hamiltonian $H(\vec q,\vec p)$. 
If one makes a change of variables $\vec q\to\vec Q(\vec q)$ (these are known as point canonical transformations in Hamiltonian mechanics), 
one can reformulate the quantum theory in terms of $\vec Q$ and its canonical conjugate $\vec P$. 
The theory is governed by the \emph{dual} Hamiltonian: 
$H(\vec q,\vec p)\to H_{\text{dual}}(\vec Q, \vec P)$.
The corresponding phase-space path integral formulation will involve integrations over paths $\vec Q(t)$ and $\vec P(t)$.
(If either $\vec Q(t)$ or $\vec P(t)$ appears at most quadratically in the phase-space action, then it can be integrated out of the path integral -- as it happens in certain cases of interest.)
In this way one can arrive at a \emph{dual} path integral formulation of the theory.

In soliton physics, a change of variables has been used to calculate transition amplitudes involving solitons \cite{Gervals1976, gervais1976point, mandelstam1975soliton, tomboulis1975canonical, christ1975quantum}.
In the statistical mechanics problem of the Coulomb gas, a quantum change of variables has been used to prove the occurrence of a mass gap \cite{Brydges} (see also the discussion in \cite{clay} \S6.6).

As an application of this method to the non-perturbative dynamics of Yang-Mills (YM), consider the Durand-Mendel result of \S\ref{cps}. Using this result, the connection on a spatial slice in the Fock-Shwinger gauge can be mapped bijectively to the magnetic field subject to the YM Bianchi identity for curvature.
The field that is canonically conjugate to the magnetic field is the \emph{dual connection} \cite{hooft} on the spatial slice \cite{halpern} (see Fig. \ref{cov}).
So the change of variables allows one to reformulate the theory in terms of the magnetic field and spatial dual connection. 
The magnetic field can be integrated out of the phase-space path integral completely (because it appears only quadratically in the phase-space action), allowing the path integral to be expressed solely in terms of the dual connection.
Now, a formulation of YM in terms of the dual connection has long been a candidate for solving the non-perturbative dynamics of YM \cite{kondo}: it is established that if the dual connection acquires a v.e.v. \cite{Seiberg, Giacomo} or a mass \cite{Tong} (\S8.3.2), \cite{clay}, then confinement follows.
\begin{figure}[ht]
    \centering
    \includegraphics[scale=0.4]{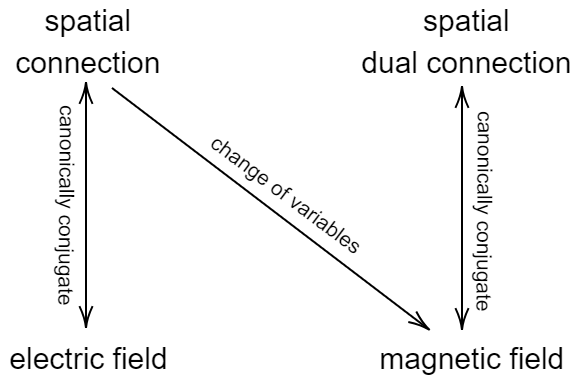}
    \caption{The chain of relations.}
    \label{cov}
\end{figure}



In quantum gravity, the effective field theory (EFT) formalism arises naturally because counter-terms to the Einstein action generate all interactions allowed by symmetry. In the EFT of gravity there are an infinite number of allowed interactions.
The space of couplings associated with these interactions is, obviously, infinite-dimensional. 
The question of the high energy behaviour of the theory comes down to the renormalization group flow of the couplings with the energy scale -- the initial conditions for this flow are determined by the requirement that at low-energies the theory reproduces general relativity, i.e., at low energies, all but Newton's constant and the cosmological constant are zero (these themselves being determined by low-energy, classical experiments).
Since Newton's constant, $G$, has negative mass dimension (which is the source of gravity's nonrenormalizability), the free-fixed-point (with all couplings equal to zero) is ultraviolet repulsive in the direction of $G$. Therefore the couplings grow as the energy increases, and perturbation theory can no longer describe the flow accurately. 

Nonrenormalizability is in itself not disastrous (in fact, from this point of view, nonrenormalizability is inconsequential for the UV finiteness of the theory), but 
what would be disastrous for the theory (in the sense that the theory cannot be physical) is if reaction rates develop singularities at finite, but very high, energies (otherwise reaction rates are finite at all energies and everything is dandy). 
It was Wienberg \cite{weinbergcritical} who first proposed this more general criteria. Notice that it is vastly more accommodating than the criteria of renormalizability. We will say that a theory satisfying Weinberg's criteria is \emph{UV finite}.
In addition, Weinberg \cite{weinberg, ambjorn, xavier} suggested that one way of, fairly surely, 
guaranteeing that this condition is met is to require that the couplings asymptotically approach fixed values at high energies -- in this case, reaction rates can, reasonably, be expected to be finite at all energies. 
This goes under the name of \emph{asymptotic safety}. 
Asymptotic safety appears to be a \emph{sufficient} condition for UV finiteness but it is certainly not a \emph{necessary} condition. In fact, even if the couplings diverged at finite, but high, energy (let alone asymptotically), it is not clear whether reactions rates would too. 

So we learn that `nonrenormalizability' cannot be conflated with the theory being unphysical; it is simply an indication that we must look beyond perturbation theory. 
What is more important than the pronouncement that gravity is nonrenormalizable is to determine the flow of the coupling constants and the reaction rates as functions of the couplings beyond perturbation theory.
What is required, then, is a formulation of quantum gravity in which non-perturbative calculations are made possible.


The result of this paper can be used to derive a dual formulation of quantum gravity as follows: one can bijectively map the \emph{spatial metric} in the normal-coordinates gauge to the spatial curvature 2-form subject to the 1\textsuperscript{st} and 2\textsuperscript{nd} Bianchi identities for curvature.
In line with the discussions above, this change of variables allows one to reformulate quantum gravity in terms of the spatial curvature and its canonical conjugate. 
The change of variables from the metric to the curvature is in analogy with the change of variables from the connection to the curvature in YM. 
(We should also point out that, as in YM, it may happen that the spatial curvature can be integrated out of the phase-space path integral in the dual formulation of quantum gravity, allowing a complete formulation in terms of its canonical conjugate.) 
The expectation that this dual formulation of quantum gravity will be non-perturbative appeals to the fact that, as we have discussed, the analogous variables change in YM provides a formulation which has strong reasons to be non-perturbative.

\subsection{Plan}

The plan of the paper is as follows. 
In \S\ref{YM} we review the Durand-Mendel result in YM. In \S\ref{MSV} we review the Müller-Schubert-van de Ven result in gravity.
In \S\ref{cl1}-\S\ref{cl3} -- which are our original contributions -- we prove the following statements respectively:
\begin{itemize}
    \item Curvatures satisfying the 1\textsuperscript{st} and 2\textsuperscript{nd} Bianchi identities for curvature are bijective with spin-connections satisfying the Fock-Shwinger gauge and torsionless conditions.
    \item Spin-connections satisfying the Fock-Shwinger gauge and torsionless conditions are bijective with vielbeins satisfying \eqref{one}-\eqref{three}.
    \item Vielbeins satisfying \eqref{one}-\eqref{three} are bijective with metrics in normal coordinates.
\end{itemize}
These statements, taken together, allow us to prove that: \emph{curvatures satisfying the 1\textsuperscript{st} and 2\textsuperscript{nd} Bianchi identities for curvature are bijective with metrics in normal coordinates}.
Fig. \ref{psumm} below summarizes the proof.
\begin{figure}[h]
    \centering
    \includegraphics[scale=0.43]{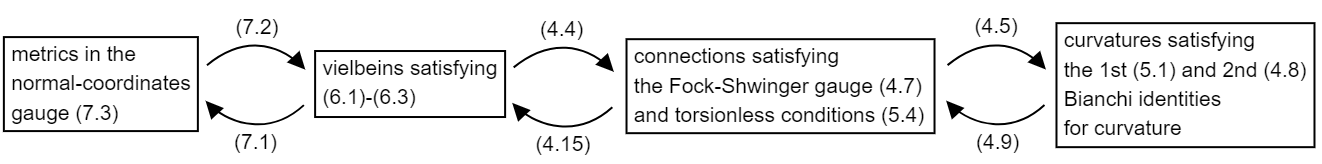}
    \caption{A summary of the proof, with references to equation numbers.}
    \label{psumm}
\end{figure}

\section{Conventions and Basics \label{conv}}
Let $x^\mu$ denote the coordinates. We adopt Greek indices for the components of tensors in the coordinate basis:
\begin{equation}
    T^{\alpha_1...\alpha_p}_{\beta_1...\beta_q}\equiv T(dx^{\alpha_1},...,dx^{\alpha_p},\partial_{\beta_1},...,\partial_{\beta_q})
\end{equation}
We adopt Wald's \cite{wald} (appendix B) conventions for the standard operations on differential forms:
\begin{alignat}{3}
&\text{exterior derivative:}\quad(d\omega)_{\mu_1...\mu_{p+1}}&&=(p+1) \partial_{[\mu_1}\omega_{\mu_2...\mu_{p+1}]} \\
&\text{wedge product:}\quad (\omega\wedge \sigma)_{\mu_1...\mu_{p+q}}&&=\frac{(p+q)!}{p!q!}\omega_{[\mu_1...\mu_p}\sigma_{\mu_{p+1}...\mu_{p+q}]} \\
&\text{interior product:}\quad (i_{X}\omega)_{\mu_2...\mu_{p}}&&=X^{\mu_1}\omega_{\mu_1\mu_2...\mu_p}
\end{alignat}
where antisymmetrization is defined in such a way that $\omega_{\mu_1...\mu_n}=\omega_{[\mu_1...\mu_n]}$.

The Lie derivative is defined as
\begin{equation}
    (\mathcal L_XT)_{\mu_1...\mu_p}=X(T_{\mu_1...\mu_p})-\sum_{i=1}^p T(\partial_{\mu_1},...,\partial_{\mu_{i-1}},[X,\partial_{\mu_i}],\partial_{\mu_{i+1}},...,\partial_{\mu_p})
\end{equation}
and satisfies
\begin{align}
    \mathcal L_X&=\{d,i_X\}. \label{Liecom} 
\end{align}
As a result,
\begin{equation}
    [\mathcal L_X,d]=0. \label{LieAncom}
\end{equation}
We define the \emph{radial vector field}
\begin{equation}
\mathfrak r\equiv x^\mu\partial_\mu~. \label{rvf}
\end{equation}
Since $[\mathfrak r,\partial_\mu]=-\partial_\mu$
\begin{equation}
(\mathcal L_{\mathfrak r} \sigma)_{\mu_1...\mu_p}=(\mathfrak r+p)\sigma_{\mu_1...\mu_p} \label{Lie}
\end{equation}
for any $p$-form $\sigma$.

If $\eta$ is a $p$-form, less singular than $|x|^{-p}$ near $x=0$, then
\begin{equation}
    \mathcal L_{\mathfrak r}\eta=0 \Longleftrightarrow \eta=0. \label{uid}
\end{equation}
This follows from the fact that the l.h.s. of \eqref{uid} is equivalent to (using \eqref{Lie})
\begin{equation}
    \frac{d}{dt}(t^p\eta_{\mu_1...\mu_p}(tx))=0
\end{equation}
which integrates to $\eta=0$ assuming $\eta$ is less singular than $|x|^{-p}$ near $x=0$.


\section{The Durand-Mendel result in Yang-Mills \label{YM}}
In this section we review the results of \cite{durand} which introduces several of the notions relevant for the latter discussions on gravity.

Connections in Yang-Mills (YM) theory are Lie algebra valued 1-forms $A$. Without loss of generality we may assume that the Lie algebra is a matrix algebra. The wedge product of matrix valued forms, of degree $p$ and $q$, is defined as
\begin{equation}                [\omega\wedge\eta]^a_b:=\omega^a_c\wedge\eta^c_b+(-)^{pq+1}\eta^a_c\wedge\omega^c_b
\end{equation}
(where $a,b$ and $c$ are the `matrix' indices).
Consider the \emph{Fock-Schwinger gauge} \cite{cronstrom, kummer}
\begin{equation}
    i_{\mathfrak r}A=0 \label{gc}
\end{equation}
where $\mathfrak r$ is the radial vector field \eqref{rvf}. 
This is a completely fixed gauge\footnote{Which means that no infinitesimal gauge transformation preserves the gauge condition. To see this, note that an infinitesimal gauge transformation $A\to A+i[A,\omega] -\frac{1}{g}d\omega$ (where the matrix $\omega(x)$ is the infinitesimal gauge transformation parameter) preserves the gauge condition iff $\frac{d}{dt}\omega(tx)=0$. But since $\omega(x)=0$ when $|x|=\infty$ (this is because bona fide gauge transformations must approach the identity at infinity), we must have $\omega(x)=0$.}.
When we map connections in this gauge to field strengths via 
\begin{equation}
F=dA+\frac{ig}{2}[A \wedge A] \label{FS}
\end{equation}
there are two questions which arise: 
\begin{enumerate}
    \item is the map injective?
    \item what is the image of the map?
\end{enumerate}
The answer to Q1 is yes. To see this, apply $i_{\mathfrak r}$ to \eqref{FS} and make use of \eqref{Liecom} and \eqref{gc}:
\begin{equation}
    i_{\mathfrak r} F=\mathcal L_{\mathfrak r} A. \label{int}
\end{equation}
Using \eqref{Lie} this becomes
\begin{equation}
x^\mu F_{\mu\nu}= (\mathfrak r+1)A_{v}
\end{equation}
which is easily integrated, assuming the connection is not singular at $x=0$, to give
\begin{equation}
    A_{\mu}(x)=\int_0^1 tx^\nu F_{\nu\mu}(tx) dt. \label{RS}
\end{equation}
(Note that $i_{\mathfrak r} A=0$ automatically from above.) To answer Q2 we first note that the Bianchi identity for curvature
\begin{equation}\label{BI}
\begin{split}
    &dF+ig[A\wedge F]=0, \\
    \text{where}\quad&A_{\mu}(x)=\int_0^1 tx^\nu F_{\nu\mu}(tx) dt
\end{split}
\end{equation}
is a necessary condition for a field strength to lie in the image. But it is also a sufficient condition, which is seen by applying $i_{\mathfrak r}$ to $\eqref{BI}$ and using \eqref{LieAncom} and \eqref{int}:
\begin{equation}\label{LeqF}
\begin{split}
&\mathcal L_{\mathfrak r}(F-dA-\frac{ig}{2}[A\wedge A] )=0, \\
\text{where}\quad&A_{\mu}(x)=\int_0^1 tx^\nu F_{\nu\mu}(tx) dt.
\end{split}
\end{equation}
From \eqref{uid}, we find that this is equivalent to 
\begin{equation}
\begin{split}
    &F=dA+\frac{ig}{2}[A\wedge A] \\
    \text{where}\quad&A_{\mu}(x)=\int_0^1 tx^\nu F_{\nu\mu}(tx) dt,
\end{split}
\end{equation}
This shows that if $F$ satisfies the Bianchi identity for curvature then it is indeed in the image of the curvature map restricted to connections in the Fock-Shwinger gauge. 
This completes the proof that the space of connections in YM satisfying the Fock-Shwinger gauge condition is bijective with the space of field strengths satisfying the Bianchi identity for curvature.

\section{The Müller-Schubert-van de Ven result in gravity \label{MSV}}
In this section we review the results of \cite{muller}, crucial for the latter developments.

\subsection{Vielbein formalism \label{gr}}
See \cite{jost} for the basics of the vielbein formalism. 

Consider the vielbein 1-forms $e^a$, and their vector duals $e_a$, satisfying
\begin{align}
    g(e_a,e_b)&=\delta_{ab} \\
    e^a(e_b)&=\delta^a_b
\end{align} 
(Note that we have taken the metric signature to be Euclidean because, ultimately, what we have in mind is to perform the change of variables on the spatial metric.)
In the coordinate basis, $e^a_\mu$ and $e_a^\mu$ are matrix inverses of each other so they satisfy
\begin{equation}
    e^a_\mu e_a^\nu=\delta^\nu_\mu
\end{equation}
We use Roman indices for the vielbein basis and Greek indices for the coordinate basis. The vielbeins obviously possess an $SO(d)$ gauge freedom (in $d$ dimensions). Assuming the torsion is zero, the spin-connection and curvature 2-form are defined as
\begin{align}
\omega^a_b &= \frac{1}{2}(-e^c\wedge i_bi_ade^c+i_bde^a-i_ade^b), \label{sc}  \\
R^a_b &= d\omega^a_b+\omega^a_c\wedge \omega^c_b  \label{c}
\end{align}
respectively, where $i_a$ is the interior product w.r.t. $e_a$. 

We should point out that \eqref{sc} is equivalent to the vanishing of the torsion:
\begin{equation}
de^a+\omega^a_b\wedge e^b=0. \label{tc}
\end{equation}
To prove this use the fact that $e^b\wedge i_b(\eta)=p\eta$ for any $p$-form $\eta$.
To see that \eqref{tc} implies \eqref{sc} use $de^a=-\omega^a_b\wedge e^b$ to simplify the r.h.s. of \eqref{sc}.

\subsection{Recovering the spin-connection from the curvature \label{ccb}}

Impose the Fock-Shwinger gauge condition on the connection:
\begin{equation}
i_{\mathfrak r}\omega^a_{b}=0.
\end{equation}
This fixes the $SO(d)$ gauge freedom completely. 
Then, in just the same way as in \S\ref{YM}, we can prove the following:
the space of spin-connections satisfying the Fock-Shwinger gauge condition is in bijection with the space of curvatures satisfying the \emph{2\textsuperscript{nd} Bianchi identity for curvature}:
\begin{equation}\label{2nd}
\begin{split}
     &dR^a_b+\omega^a_c\wedge R^c_b-R^a_c\wedge \omega^c_b=0, \\
     \text{where}\quad&\omega^a_{b\mu}(x)=\int_0^1 tx^\nu R^a_{b\nu\mu}(tx) dt.
\end{split}
\end{equation}
The forward map for this bijection is given by \eqref{c}; and the reverse map is given by
\begin{equation}
\omega^a_{b\mu}(x)=\int_0^1 tx^\nu R^a_{b\nu\mu}(tx) dt. \label{rm}
\end{equation}


\subsection{Recovering the vielbein from the spin-connection \label{cvc}}

Applying $i_{\mathfrak r}$ to \eqref{tc} gives
\begin{equation}
-di_{\mathfrak r}e^a+(\mathfrak r+1) e^a-\omega^a_b \wedge i_{\mathfrak r} e^b=0. \label{10}
\end{equation}
Assume we are working in \emph{`normal-coordinates'}, i.e. a coordinate system in which
\begin{equation}
    x^\mu g_{\mu\nu}(x)=x^\nu \label{cs}
\end{equation}
(note that this implies $g_{\mu\nu}(0)=\delta_{\mu\nu}$). As a result,
\begin{align}
x^\mu e^a_\mu(x)&=x^a \\
x^a e^a_\mu(x)&=x^\mu
\end{align}
(note that this implies $e^a_\mu(0)=\delta^a_\mu$).
So \eqref{10} becomes
\begin{equation}
(\mathfrak r +1)e^a_\mu=\delta^a_\mu+\omega^a_{b\mu}x^b 
\end{equation}
which is easily integrated to give
\begin{equation}
e^a_\mu(x)=\delta^a_\mu+\int_0^1 \omega^a_{b\mu}(tx) tx^b dt.  \label{ewc}
\end{equation}
Using \eqref{rm} then gives
\begin{equation}
e^a_\mu(x)=\delta^a_\mu+\int_0^1\int_0^1 t_1t_2^2x^b x^\nu R^a_{b\nu\mu}(t_1t_2x) dt_1 dt_2 
\end{equation}
Fixing $t_1t_2$ and doing the $t_2$ integral first gives, finally,
\begin{equation}
e^a_\mu(x)=\delta^a_\mu+\int_0^1 t(1-t)x^b x^\nu R^a_{b\nu\mu}(tx) dt. \label{lrm}
\end{equation}
Note that $x^\mu e^a_\mu(x)=x^a$ and $x^a e^a_\mu(x)=x^\mu$ follow automatically from this.

\section{Bijection between the curvature and spin-connection \label{cl1}}

\begin{theorem}
    Curvatures satisfying the 1\textsuperscript{st} and 2\textsuperscript{nd} Bianchi identities for curvature are bijective with spin-connections satisfying the Fock-Shwinger gauge and torsionless conditions. \label{th1}
\end{theorem}
\begin{proof}
We pointed out in \S\ref{ccb} that the space of curvatures satisfying the 2\textsuperscript{nd} Bianchi identity for curvature is in bijection with the space of connections satisfying $i_{\mathfrak r}\omega^a_b=0$. 
We now place the \emph{1\textsuperscript{st} Bianchi identity for curvature} 
\begin{equation}\label{1st}
\begin{split}
    &R^a_b\wedge e^b=0, \\
    \text{where}\quad&e^a_\mu(x)=\delta^a_\mu+\int_0^1 t(1-t)x^b x^\nu R^a_{b\nu\mu}(tx) dt
\end{split}
\end{equation}
as an additional restriction on the space of curvatures. 
Via the bijection, this descends into an extra restriction on the space of connections:
\begin{equation}
\begin{split}
    &R^a_b\wedge e^b=0, \\
    \text{where}\quad &R^a_b = d\omega^a_b+\omega^a_c\wedge \omega^c_b \\
    \text{and}\quad&e^a_\mu(x)=\delta^a_\mu+\int_0^1 \omega^a_{b\mu}(tx) tx^b dt.
\end{split}
\end{equation}
Applying $i_{\mathfrak r}$ to this and manipulating gives
\begin{equation}
\begin{split}
&\mathcal L_{\mathfrak r}(de^a+\omega^a_b \wedge e^b)=0,\\
\text{where}\quad &e^a_\mu(x)=\delta^a_\mu+\int_0^1 \omega^a_{b\mu}(tx) tx^b dt
\end{split}
\end{equation}
which implies (using \eqref{uid}) that
\begin{equation}\label{tsc}
\begin{split}
     &de^a+\omega^a_b \wedge e^b=0,\\
     \text{where}\quad&e^a_\mu(x)=\delta^a_\mu+\int_0^1 \omega^a_{b\mu}(tx) tx^b dt. 
\end{split}
\end{equation} 
We will refer to \eqref{tsc} as the \emph{torsionless conditon} on connections (c.f. \eqref{tc}).

\end{proof}

\section{Bijection between the spin-connection and vielbein \label{cl2}}

\begin{theorem}
    Spin-connections satisfying the Fock-Shwinger gauge and torsionless conditions are bijective with vielbeins satisfying \eqref{one}-\eqref{three}.\label{th2}
\end{theorem}

\begin{proof}
Consider the space of vielbeins satisfying 
\begin{align}
x^\mu e^a_\mu(x)&=x^a \label{one} \\
x^a e^a_\mu(x)&=x^\mu \label{two} \\
i_a \mathcal L_{\mathfrak r} e^b&=i_b \mathcal L_{\mathfrak r} e^a \label{three}
\end{align} 
We leave it as an exercise to the reader to show that connections arising from such vielbeins satisfy the Fock-Shwinger gauge and torsionless conditions automatically. 

Conversely, consider the image of connections satisfying the Fock-Shwinger gauge and torsionless conditions under the map \eqref{ewc}. 
It is immediately clear that such vielbeins satisfy \eqref{one}-\eqref{two}. But, in addition, such vielbeins satisfy \eqref{three}, which can be proved as follows: 
we pointed out, in \S\ref{gr}, that \eqref{tc} is equivalent to \eqref{sc}. This means that the torsionless condition is equivalent to 
\begin{equation}\label{yuz}
\begin{split}
        &\omega^a_b = \frac{1}{2}(-e^c\wedge i_bi_ade^c+i_bde^a-i_ade^b)\\
        \text{where}\quad &e^a_\mu(x)=\delta^a_\mu+\int_0^1 \omega^a_{b\mu}(tx) tx^b dt.
\end{split}
\end{equation}
Applying $i_{\mathfrak r}$ to this and using $x^a \wedge e^a = x^\mu dx^\mu$ shows that \eqref{three} is indeed satisfied.

From \eqref{yuz}, it is clear that mapping a connection -- satisfying the torsionless condition -- to a vielbein (via \eqref{ewc}) and back is equivalent to the identity operation. 
\end{proof}

\section{Bijection between the vielbein and metric \label{cl3}}

\begin{theorem}
    Vielbeins satisfying \eqref{one}-\eqref{three} are bijective with metrics in normal coordinates. \label{th3}
\end{theorem}
\begin{proof}
The metric is obtained from the vielbein using
\begin{equation}
    g_{\mu\nu}=e^a_\mu e^a_\nu \label{me}
\end{equation}
Conversely, to obtain the vielbein from the metric,
we differentiate \eqref{me} and use \eqref{three} (which, using \eqref{Lie}, can be shown to be equivalent to $e^a_{[\mu
}(tx)\frac{d}{dt}e^a_{\nu]}(tx)=0$, and which completely fixes the $SO(d)$ gauge freedom the vielbeins possess):
\begin{equation}
\frac{d}{dt} e^a_\mu(tx)=\frac{1}{2}e^\nu_a(tx) \frac{d}{dt} g_{\nu\mu}(tx). \label{egc}
\end{equation}
Given the metric, this can be integrated uniquely for the vielbeins (a fact essentially guaranteed by the Picard-Lindelöf theorem \cite{Teschl} (\S2.2)) using the initial condition $e^a_\mu(0)=\delta^a_\mu$ (note that this initial condition is consistent with \eqref{one}-\eqref{two}). This is the forward map, taking metrics to vielbeins. 

Metrics in normal coordinates satisfy (see \eqref{cs})
\begin{equation}
x^\mu g_{\mu\nu}(x)=x^\nu \label{nc}
\end{equation}
From \eqref{me} it is clear that vielbeins satisfying \eqref{one}-\eqref{two} map to metrics in normal coordinates.
The question is whether metrics in normal coordinates map to vielbeins satisfying \eqref{one}-\eqref{three}. 
\eqref{three} is easily shown by multiplying \eqref{egc} by $e_b^\mu(tx)$ and using the symmetry of the metric. 

From \eqref{egc} we have
\begin{equation}
\frac{d}{dt}(tx^\mu e^a_\mu(tx))=x^\mu e^a_\mu(tx)
\end{equation}
which implies 
\begin{equation}
    x^\mu e^a_\mu(tx)=\text{const.}=x^a \label{tx}
\end{equation}
proving \eqref{one}. 

\eqref{tx} implies $x^ae^\mu_a(x)=x^\mu$, which together with \eqref{egc} gives 
\begin{equation}
\frac{d}{dt}(tx^ae^a_\mu(tx))=x^ae^a_\mu(tx)
\end{equation}
which implies 
\begin{equation}
    x^ae^a_\mu(tx)=\text{const.}=x^\mu
\end{equation}
proving \eqref{two}. 

We leave it as an simple exercise to the reader to convince themselves that \eqref{me} and \eqref{egc} are indeed inverses of each other when restricted to the subspaces in question.
\end{proof}

\section{Conclusions \label{conc}}

Combining Theorems \ref{th1}, \ref{th2} and \ref{th3} finally proves that
\begin{theorem}
The space of curvatures satisfying the 1\textsuperscript{st} and 2\textsuperscript{nd} Bianchi identities for curvature is bijective with the space of metrics in normal coordinates.
\end{theorem}
The proof is schematically represented in Fig \ref{psumm}.

\acknowledgments

The author thanks the Rudolf Peierls Centre for their hospitality.

For the purpose of Open Access, the author
has applied a CC BY public copyright licence to any Author Accepted Manuscript version
arising from this submission.



\bibliographystyle{JHEP} 
 
\bibliography{biblio} 

\end{document}